\documentclass[acmsmall,nonacm]{acmart}

\usepackage{graphicx} 

\usepackage[T1]{fontenc}
\usepackage{amsfonts}
\usepackage{graphicx}
\usepackage{epstopdf}
\usepackage{enumitem}
\usepackage{algorithmic}
\ifpdf
  \DeclareGraphicsExtensions{.eps,.pdf,.png,.jpg}
\else
  \DeclareGraphicsExtensions{.eps}
\fi

\usepackage{amsopn}

\usepackage{xspace}
\newtheorem{algo}{Algorithm}

\newcommand{\ALG}{\textrm{HAR} \xspace}
\newcommand{\OPT}{\textrm{OPT} \xspace}

\begin{document}

\title{
  The Harmonic Policy for Online Buffer Sharing is $(2 + \ln n)$-Competitive: A Simple Proof
}


\author{Vamsi Addanki}
\affiliation{%
  \institution{Purdue University}
  \city{West Lafayette}
  \state{Indiana}
  \country{USA}
}

\author{Julien Dallot}
\affiliation{%
  \institution{TU Berlin}
  \city{Berlin}
  \state{Berlin}
  \country{Germany}
}

\author{Leon Kellerhals}
\affiliation{%
  \institution{TU Clausthal}
  \city{Clausthal}
  \country{Germany}
}

\author{Maciej Pacut}
\affiliation{%
  \institution{TU Berlin}
  \city{Berlin}
  \state{Berlin}
  \country{Germany}
}

\author{Stefan Schmid}
\affiliation{%
  \institution{TU Berlin}
  \city{Berlin}
  \state{Berlin}
  \country{Germany}
}

\date{}

\begin{abstract}
  The problem of online buffer sharing is expressed as follows.
  A switch with $n$ output ports receives a stream of incoming packets.
  When an incoming packet is accepted by the switch, it is stored in a shared buffer of capacity $B$ common to all packets and awaits its transmission through its corresponding output port determined by its destination.
  Each output port transmits one packet per time unit.
  The problem is to find an algorithm for the switch to accept or reject a packet upon its arrival in order to maximize the total number of  transmitted packets.

  Building on the work of Kesselman et al. (STOC 2001) on split buffer sharing, Kesselman and Mansour (TCS 2004) considered the problem of online buffer sharing which models most deployed internet switches.
  In their work, they presented the Harmonic policy and proved that it is $(2 + \ln n)$-competitive, which is the best known competitive ratio for this problem.
  The Harmonic policy unfortunately saw less practical relevance as it performs $n$ threshold checks per packets which is deemed costly in practice, especially on network switches processing multiple terabits of packets per second.
  While the Harmonic policy is elegant, the original proof is also rather complex and involves a lengthy matching routine along with multiple intermediary results.

  This note presents a simplified Harmonic policy, both in terms of implementation and proof.
  First, we show that the Harmonic policy can be implemented with a constant number of threshold checks per packet, matching the widely deployed \emph{Dynamic Threshold} policy.
  Second, we present a simple proof that shows the Harmonic policy is $(2 + \ln n)$-competitive.
  In contrast to the original proof, the current proof is direct and relies on a 3-partitioning of the packets.
\end{abstract}

\maketitle

\section{Introduction}
In the online buffer sharing problem, we consider a network switch with $n$ output ports labeled $1 \dots n$ and a shared buffer which can hold at most $B$ packets.
A stream of packets $\sigma = p_1, p_2, \dots$ arrives in the switch.
Each packet $p_i$ arrives at some time $t_i$ and must be transmitted through a specific output port $q_i \in 1 \dots n$.
We assume that the packets can be totally ordered by their arrival time --- if two packets arrive at the same time, we assume that we break the tie arbitrarily.

When a packet $p_i$ arrives, the switch decides to either accept or reject it.
If $p_i$ is accepted, it is stored in the buffer and queued for transmission through the output port $q_i$ in a first-in-first-out (FIFO) manner.
Each output port transmits up to one packet at each integer time $t = 1, 2, \dots$; the output port does not transmit anything if no packet is queued.

The task is to design an \emph{acceptance policy} for the switch that decides whether to accept or reject each incoming packet.
The objective is to maximize the number of accepted packets while respecting the shared buffer capacity at all time.
The acceptance policy must be \emph{online}, which means that the decision to accept or reject a packet must be made at the time of its arrival without knowledge of future packets.

Several buffer sharing policies have been established over the years.
Among the most relevant are the Sharing with Maximum Queue Lengths (SMXQ)~\cite{1094076} and Dynamic Thresholds~\cite{7859368} algorithms, with the latter being the most used in practice.
Both have a competitive ratio no better than $O(n)$, i.e., there exists a sequence of packets such that an optimal, offline acceptance policy accepts $O(n)$ more packets.

Following a series of results on switch buffer sharing~\cite{10.1145/380752.380847, 414642}, Kesselman and Mansour~\cite{KESSELMAN2004161} present the Harmonic policy for online buffer sharing in 2004.
The main novelty of the Harmonic policy was to consider the queues sorted by their current size: it assigns roughly a $1/i$ fraction of the memory to the $i$-th largest queue, resulting in a competitive ratio of $(2+\ln n)$.
Their policy works as follows.

\begin{algo}[The Harmonic Policy]
  \label{def:harmonic}
  Accept the incoming packet if, after accepting it, the following would hold for each $i \in 1 \dots n$: the total occupancy of the $i$ queues with greatest occupancies is at most $\frac{B}{1 + \ln n} \cdot \sum_{k \in 1 \dots i} \frac{1}{k}$.
\end{algo}

The Harmonic policy is the best known algorithm with respect to the competitive ratio metric --- Kesselman and Mansour also showed a lower bound of $\Omega(\log n / \log \log n)$ for any online, deterministic policy, showing that Harmonic is close to optimal.
Even though the alternative algorithms all have a worst-case competitive ratio of $O(n)$, the Harmonic policy however saw little practical relevance in modern network switches.
This is arguably due to the initial presentation of the policy~\cite{KESSELMAN2004161}, which suggests that one must maintain a list of the output ports by queue occupancy and must perform $n$ threshold checks upon each packet arrival, which is deemed too costly in practice.
Moreover, the original proof of Harmonic's competitive ratio is rather complex, involving a lengthy matching routine along with multiple intermediary results, which makes it harder for practicioners to adopt and understand the algorithm.
Nowadays, buffer management algorithms are revisited through the lens of \emph{learning-augmented} algorithms~\cite{295535, 9796967, 9796784}, where the acceptance policy uses machine-learned predictions to improve its performances.
Most of the algorithms are based on one of the established algorithms and often use competitive analysis as a key metric to evaluate the prediction quality~\cite{LykourisV21, 295535}; The Harmonic algorithm therefore plays a central role as both a reference and a base to build new learning-augmented algorithms.

In this work, we first want to point out a quick observation that the Harmonic algorithm can be implemented to use only a constant number of operations per packet.
Secondly, and mainly, we present a significantly simpler proof for its competitive ratio.
We hope that this simpler proof makes Harmonic more accessible and possibly leads the way to more competitive learning-augmented algorithms for buffer sharing.

\section{Efficiency of the Harmonic Policy}

We briefly show that the Harmonic policy can be implemented very efficiently.
Central to the Harmonic policy are the \emph{thresholds}, which are proportional to the harmonic series.

\begin{definition}[Thresholds]
  For all $k = 1 \dots n$, we define the threshold $T_k = \frac{B}{1 + \ln n} \cdot \frac{1}{k}$.
\end{definition}

\begin{algo}[The Modified Harmonic Policy]
  \label{def:modified_harmonic}
  When packet $p_i$ arrives, compute the smallest $T_k$ such that the occupancy of queue $q_i$ is strictly lower than $T_k$.
  Accept $p_i$ if and only if accepting it results in at most $k$ queues with occupancy no lower than $T_k$.
\end{algo}

For each queue, the smallest threshold $T_k$ greater than its occupancy can be tracked in constant time per packet: each time a new packet is accepted or transmitted for a given queue, compare the new queue's occupancy with the current threshold $T_{k}$ and update it to $T_{k+1}$ or $T_{k-1}$ accordingly.
Finally, it is possible to track the number of queues no lower than a given threshold in constant time per packet.
For this, hold a list where values are associated with the thresholds, $T_1, T_2, \dots, T_n$ from left to right, and maintain the cumulative number of queues whose occupancy is above each threshold.
When a packet $p_i$ arrives, both the list update and the threshold check can be performed in constant time (for the update: access the value of the previous threshold and increment the value on its left, for the threshold check: access the value and compare it with the corresponding threshold).
Those structures can also be maintained efficiently when packets are transmitted from the queues: the number of operations is then in $O(n)$, a constant operation per queue, which is comparable with Dynamic Threshold~\cite{7859368} where the algorithm keeps track of all queue occupancies.
All those changes require only two additional integer variables per queue and maximum two increments or decrements operations per packet arrival. 

\section{A Simpler Proof for the Competitive Ratio}

We next give a simpler proof for the following theorem, originally proven by Kesselman and Mansour~\cite{KESSELMAN2004161}.

\begin{theorem}
  \label{th:log_competitive_ratio}
  The Modified Harmonic Policy is $(2 + \ln n)$-competitive.
\end{theorem}

\begin{proof}
  We refer to the Modified Harmonic policy as Harmonic, $\ALG$ is the set of packets accepted by Harmonic and $\OPT$ is the set of packets accepted by an offline, optimal policy.

  \noindent
  We partition $\OPT$ in three subsets $A$, $B$ and $C$ as follows:
  \begin{enumerate}
  \item 
    $A = \OPT \cap \ALG$
  \item 
    $B$ contains the packets $p \notin \ALG$ where $p$'s associated queue is strictly higher in $\ALG$'s cache than in $\OPT$'s cache at the time $p$ arrives.
  \item 
    $C = \OPT \setminus (A \cup B)$ 
  \end{enumerate}

  We will first prove that $|\ALG| \ge |A| + |B|$ by providing an injective mapping from $A \cup B$ to $\ALG$.
  To obtain the desired mapping, we will construct injective sub-mappings restricted to the packets with output port $q$ for each $q \in 1 \dots n$ and then merge all the submappings.
  The sub-mappings are constructed iteratively as packets arrive.
  When packet $p_i \in A \cup B$ with output port $q \in 1 \dots n$ arrives, we extend the current sub-mapping as follows:
  \begin{itemize}
  \item
    If $p_{i} \in A$, then map $p_{i}$ to itself in $\ALG$.
  \item
    If $p_{i} \in B$, then find a packet $p \in \ALG$ with the same output port $q$ that already arrived and to which no packet was mapped yet, and map $p_i$ to $p$.
  \end{itemize}

  We are left to prove that there always exists an unmapped packet like described in the second mapping rule.
  For a given packet $p_i$, we call $\text{Occ}(i, q, \ALG)$ the occupancy of queue $q$ at time $t_i$ in $\ALG$ (likewise for $\OPT$) and define:
  \begin{align*}
    g_i(q) = \max \left\{\text{Occ}(i, q, \ALG) - \text{Occ}(i, q, \OPT), 0\right\}.
  \end{align*}
  We also define $u_i(q)$ to be the number of packets of $\ALG$ in queue $q$ without antecedent in the sub-mapping at time $t_i$.
  We prove by induction on the arriving packets that
  \begin{align*}
    \forall q, \forall i,  \quad g_i(q) = u_i(q),
  \end{align*}
  which would directly imply that the described mapping is possible as then $p_{i+1} \in B \implies g_i(q) > 0 \implies u_i(q) > 0$.
  The base case is clearly true at time $0$ for all $q$.
  Assuming that this holds for some $i$, we prove that it also holds for $i+1$ with a case distinction on the incoming packet $p_{i+1}$.
  If (1) $p_{i+1} \in B$ then both $g_{i+1}(q)$ and $u_{i+1}(q)$ are decremented by $1$.
  Otherwise when (2) $p_{i+1} \in A$ then both $g_{i+1}(q)$ and $u_{i+1}(q)$ stay unchanged.
  Otherwise if (3) $p_{i+1} \in C$ then both $g_{i+1}(q)$ and $u_{i+1}(q)$ stay unchanged.
  Finally, if (4) $p_{i+1} \in \ALG \setminus \OPT$ then both $g_{i+1}(q)$ and $u_{i+1}(q)$ are incremented by $1$.

  We constructed an injective mapping from $A \cup B$ to $\ALG$ restricted to the packets with output port $q$; we obtain an injective mapping from $A \cup B$ to $\ALG$ by merging the sub-mappings for each output port $q = 1 \dots n$, proving that $|\ALG| \ge |A| + |B|$.\\

  We then prove that $(1 + \ln n) \cdot |\ALG| \ge |C|$.
  We prove it by exhibiting a \emph{matching} between the packets of $C$ and those of $\ALG$ such that:
  \begin{enumerate}
  \item
    each packet of $C$ is matched to exactly one packet of $\ALG$
  \item 
    each packet of $\ALG$ is matched to at most $1 + \ln n$ packets of $C$.
  \end{enumerate}

  We construct the matching iteratively as packets of $C$ arrive.
  We prove that, when \OPT accepts a packet $p \in C$, there is always a packet in $\ALG$'s buffer available for matching that will be drained earlier, that is, matched with strictly less than $1 + \ln n$ packets and with a lower waiting queue than $p$.
  We do this by contradiction: let $p \in C$ be the first packet that is impossible to match without violating the wanted condition on the matching.
  We note $T_k$ the threshold that triggered the rejection of $p$ by Harmonic, that is, Harnomic's buffer would contain strictly more than $k$ queues with occupancy no lower than $T_k$ if $p$ had been accepted.
  When \OPT accepts $p$, its waiting queue is greater than $T_k$ as $p \notin B$.
  Thus by assumption, all the packets in $\ALG$'s current buffer with position no greater than $T_k$ are already matched with $1 + \ln n$ packets in $C$.
  Notice that the packets in $\ALG$'s buffer are matched to packets that are currently in $\OPT$'s buffer --- this is true since, whenever some packet $p^{\prime} \in C$ is matched, its position in $\OPT$'s queue is higher than that of its mate packet in $\ALG$'s queue and it will therefore be drained later.
  Harmonic's buffer contains at least $k \cdot \left\lceil T_k \right\rceil \ge \frac{B}{1 + \ln n}$ packets.
  Each of those packets are matched with $1 + \ln n$ packets of $C$ in $\OPT$'s current buffer, which thus contains at least $B$ packets without counting $p$: $\OPT$ cannot have accepted $p$, a contradiction.

  Summing the two results obtained so far gives $(2 + \ln n) \cdot |\ALG| \ge |A| + |B| + |C| = |\OPT|$, proving the claim.
\end{proof}

\paragraph*{Acknowledgements}
Research supported by the German Research Foundation (DFG), grant 470029389 (FlexNets), 2021-2025.

\bibliographystyle{ACM-Reference-Format}
\bibliography{references}

\end{document}